\PassOptionsToPackage{final}{microtype}
\PassOptionsToPackage{final}{hyperref}
\documentclass[format=sigconf]{acmart}
\copyrightyear{2021}
\acmYear{2021}
\setcopyright{acmlicensed}\acmConference[ISSAC '21]{Proceedings of the 2021 International Symposium on Symbolic and Algebraic Computation}{July 18--23, 2021}{Virtual Event, Russian Federation}
\acmBooktitle{Proceedings of the 2021 International Symposium on Symbolic and Algebraic Computation (ISSAC '21), July 18--23, 2021, Virtual Event, Russian Federation}
\acmPrice{15.00}
\acmDOI{10.1145/3452143.3465522}
\acmISBN{978-1-4503-8382-0/21/07}

\usepackage[utf8]{inputenc}
\usepackage[T1]{fontenc}
\usepackage{amsmath,amsthm}
\usepackage{xspace}
\DeclareFontFamily{U}{mathb}{\hyphenchar\font45}
\DeclareFontShape{U}{mathb}{m}{n}{ <-6> mathb5 <6-7> mathb6 <7-8>
  mathb7 <8-9> mathb8 <9-10> mathb9 <10-12> mathb10 <12-> mathb12 }{}
\DeclareSymbolFont{mathb}{U}{mathb}{m}{n}
\DeclareMathSymbol{\prec}{\mathrel}{mathb}{"A0}
\DeclareMathSymbol{\succ}{\mathrel}{mathb}{"A1}
\DeclareMathSymbol{\preceq}{\mathrel}{mathb}{"A8}
\DeclareMathSymbol{\succeq}{\mathrel}{mathb}{"A9}
\DeclareMathSymbol{\precneq}{\mathrel}{mathb}{"AC}
\DeclareMathSymbol{\succneq}{\mathrel}{mathb}{"AD}

\usepackage{color}
\definecolor{gray}{gray}{0.4}

\usepackage{amscd}
\usepackage{natbib}
\usepackage{enumerate}
\usepackage[algoruled,vlined,english,linesnumbered]{algorithm2e}
\SetArgSty{textup}
\SetKwInOut{Input}{Input}\SetKwInOut{Output}{Output}
\SetKwIF{If}{ElseIf}{Else}{if}{:}{elif}{else:}{}%
\SetKw{KwTo}{in}\SetKwFor{For}{for}{\string:}{}%
\SetKwFor{While}{while}{:}{fintq}%
\AlgoDontDisplayBlockMarkers
\SetAlgoNoEnd
\newcommand{\indictype}{}
\makeatletter
\newcommand{\removelatexerror}{\let\@latex@error\@gobble}
\makeatother

\usepackage[font=small]{caption}
\usepackage{subcaption}

\usepackage[english]{babel}
\usepackage[shortlabels,inline]{enumitem} 

\usepackage{tikz}

\usepackage{fixme}
\fxsetup{author={T},envlayout={color},%
  layout={footnote,marginclue},theme=color}

\FXRegisterAuthor{mf}{amf}{M}

\usepackage[mathscr]{euscript}

\usepackage{array}

\usepackage{mathtools}
\mathtoolsset{showonlyrefs,showmanualtags}

\usepackage{tabu}
\usepackage{multirow}
\usepackage{booktabs}

\usepackage{siunitx}

\newtheorem{theorem}{Theorem}[section]
\newtheorem{lemma}[theorem]{Lemma}
\newtheorem{proposition}[theorem]{Proposition}

\newtheorem{definition}[theorem]{Definition}
     
\newtheorem{example}[theorem]{Example}

\newcommand{\LM}{\mathrm{lm}}
\newcommand{\LT}{\mathrm{lt}}

\newcommand{\ZZ}{\mathbb{Z}}

\newcommand{\NN}{\mathbb{N}}

\newcommand{\divides}{\mid}
\newcommand{\LC}{\mathrm{lc}}

\newcommand{\lm}{\LM}
\newcommand{\lt}{\LT}
\newcommand{\lc}{\LC}

\newcommand{\lcmLM}{\textup{lcmlm}}

\newcommand{\lcmlm}{\lcmLM}
\newcommand{\lcmlt}{\textup{lcmlt}}

\newcommand{\lcm}{\mathrm{lcm}}

\newcommand{\intint}[1]{\llbracket #1 \rrbracket}

\renewcommand{\bar}[1]{\overline{#1}}
\newcommand{\npols}{n_\mathrm{polys}}

\newcommand{\nvars}{n_\mathrm{vars}}

\newcommand{\II}{\mathcal{I}}
\newcommand{\GG}{\mathcal{G}}

\newcommand{\bfu}{\mathbf{u}}
\newcommand{\bfe}{\mathbf{e}}
\newcommand{\bff}{\mathbf{f}}
\newcommand{\bfp}{\mathbf{p}}
\newcommand{\bfg}{\mathbf{g}}
\newcommand{\bfh}{\mathbf{h}}
\newcommand{\bfs}{\mathbf{s}}
\newcommand{\bft}{\mathbf{t}}
\newcommand{\bfT}{\mathbf{T}}

\newcommand{\bfz}{\mathbf{z}}

\newcommand{\sig}{\textup{sig}}
\newcommand{\redsig}{\textup{sig}}
\newcommand{\redpol}{\textup{s}}

\newcommand{\sigS}{\textup{sig}}
\newcommand{\lmS}{\textup{mdeg}}
\newcommand{\ltS}{\textup{tdeg}}

\newcommand{\sigG}{\textup{sigG}}
\newcommand{\sigGcomb}{\textup{sigG-Comb}}

\newcommand{\mm}[1]{\textup{m}_{#1}}

\newcommand{\cc}[1]{\textup{c}_{#1}}

\newcommand{\SPol}{\textup{S-Pol}}
\newcommand{\GPol}{\textup{G-Pol}}
\newcommand{\SGPol}{\textup{SG-Pol}}

\newcommand{\sppair}[1]{(#1,\bar{#1})}

\newcommand{\Syz}{\mathrm{Syz}}

\newcommand{\Mon}{\mathrm{Mon}}
\newcommand{\Ter}{\mathrm{Ter}}

\renewcommand{\AA}{A}

\everypar{\looseness=-1}

\begin{document}

\title{On Two Signature Variants Of Buchberger's Algorithm\\
  Over Principal Ideal Domains}

\author{Maria Francis}
\affiliation{
  \institution{Indian Institute of Technology Hyderabad}
  \city{Hyderabad}
  \country{India}  
}
\email{mariaf@iith.ac.in}

\author{Thibaut Verron}
\affiliation{
  \institution{Institute for Algebra / Johannes Kepler University}
  \city{Linz}
  \country{Austria}  
}
\email{thibaut.verron@jku.at}

\thanks{T.\ Verron was supported by the Austrian FWF grant P31571-N32.}



\begin{abstract}
  Signature-based algorithms have brought large improvements in the performances of Gröbner bases algorithms for polynomial systems over fields.
  Furthermore, they yield additional data which can be used, for example, to compute the module of syzygies of an ideal or to compute coefficients in terms of the input generators.

  In this paper, we examine two variants of Buchberger's algorithm to compute Gröbner bases over principal ideal domains, with the addition of signatures.
  The first one is adapted from Kandri-Rody and Kapur's algorithm~\cite{Kandri-Rody-Kapur}, whereas the second one uses the ideas developed in the algorithms by L.~Pan~\cite{Pan:Dbases}
  and D.~Lichtblau~\cite{Lichtblau}.
  The differences in constructions between the algorithms entail differences in the operations which are compatible with the signatures, and in the criteria which can be used to discard elements.

  We prove that both algorithms are correct and discuss their relative performances in a prototype implementation in Magma.
\end{abstract}



\vspace{-1.5mm}


\maketitle

\addtolength\abovedisplayskip{-0.2\baselineskip}%
\addtolength\belowdisplayskip{-0.2\baselineskip}%


\section{Introduction}
\label{sec:Introduction-1}

Gröbner bases over fields, introduced by Buchberger \cite{Buchberger:1965:thesis}, is a fundamental tool in  computational ideal theory and algebraic geometry. Very early on, several approaches were proposed to extend the algorithmic theory of Gröbner bases to polynomial rings over rings, a summary of which can be found in~\cite{Adams:1994:introtogrobnerbasis, Becker:1993:grobnerbasistext}.
Ideals in polynomial rings over rings have several applications, for instance in number theory~\cite{Lichtblau:applications}, in coding theory~\cite{Norton-2001-strong-groebner-bases, Norton-2003-cyclic-codes-minimal} or in 
cryptography
~\citep{Lyubashevsky:2006:Ideallatticefirstdef, FrancisDukkipati:2014:Hash}. 

There are two ways to define Gröbner bases (GB) over rings, namely weak and strong Gröbner bases, corresponding to two notions of reductions.
Of the two, strong GB and reductions are the most similar to fields and the ones we consider in this work. It allows to efficiently compute the normal form of an element and over principal ideal domains (PID's), all ideals admit a strong GB.
Different algorithms have been proposed for computing strong bases over PID's~\cite{Moller:1988:grobnerrings2,Pan:Dbases} and over Euclidean domains~\cite{Kandri-Rody-Kapur, Lichtblau, Eder:2017:EuclideanRings, EderPopsecu:Standardbases:2018}.

Buchberger's original algorithm for computing Gröbner bases over fields proceeds by computing and reducing S-polynomials.
Over rings, the computation of Gröbner bases additionally requires to compute so-called G-polynomials, that is, combinations of polynomials which use Bézout coefficients to make the leading coefficient as small as possible.
Kandri-Rody and Kapur's algorithm~\cite{Kandri-Rody-Kapur} was designed for Euclidean domains but works without any modification over PIDs; it proceeds by computing, for each pair of elements, both their S- and G-polynomials, and adding them to the queue for later processing.
Pan's algorithm~\cite{Pan:Dbases} for PIDs, later refined by Lichtblau~\cite{Lichtblau} for Euclidean domains, observes that for each pair, only one polynomial, S- or G-, is required.

Over fields, it was rapidly noticed that many of the reductions in Buchberger's algorithm are useless, \emph{i.e.}, they eventually reach 0 and are discarded. Optimizations of Buchberger's algorithm that started with Buchberger~\cite{buchberger1979criterion} have focused on how to detect these useless reductions beforehand~\cite{moller1992grobner}.
A breakthrough came in the early 2000s with the class of so-called signature-based algorithms such as F5~\cite{Faugere:2002:F5} and later GVW~\cite{Gao-2015-new-framework-for}. A comprehensive survey of signature-based algorithms can be found in \cite{eder:2017:survey}.
These algorithms keep track of the signature of each computed polynomial, that is, the leading term of a representation of the polynomial in terms of the generators of the ideal.
This information can be used to detect reductions to 0, and avoid redundant computations.

Furthermore, the computation of a Gröbner basis with signatures allows to recover the coefficients of the elements of the basis in terms of the generators, and to compute the module of syzygies of those generators, without the extra cost of module computations or additional variables~\cite{Gao-2015-new-framework-for}.

The natural next step is to see whether these signature-based techniques can be generalized to Gröbner basis algorithms over rings. In this direction, a hybrid algorithm was presented in~\cite{Eder:2017:EuclideanRings} that added signatures to a modified version of Kandri-Rody and Kapur's algorithm. The authors showed with a counter-example that implementing \emph{totally ordered} signatures for rings cannot ensure that the signatures will never decrease/drop during the course of computing the strong GB, which is the key invariant of most signature-based algorithms.
The signature-based techniques of \cite{Eder:2017:EuclideanRings} could however be used as an efficient preprocessing step to speed up the computations, falling back to the classical techniques when a drop in signature is detected.

In~\citep{FV2018}, the authors described an algorithm that computes a weak Gröbner basis with signatures, over PID's, without any signature drop, by using a partial order on the signatures.
However, as is the case without signatures, this algorithm is mostly of theoretical interest, as the exponential cost of computing weak S-polynomials makes it impractical for all but the smallest examples.
In this work, we use similar constructions to adapt Kandri-Rody and Kapur's algorithm and Pan/Lichtblau's algorithm to the computation of signature Gröbner bases.
For that purpose, we need two constructions: a restriction on the construction of G-polynomials ensuring that we can keep track of their signatures without discarding any; and an analogous construction using Bézout coefficients to obtain elements with small signatures.
In the case of Kandri-Rody and Kapur's algorithm, we prove that the powerful cover criterion described in~\cite{Gao:2010:GVW} can be applied to eliminate some S-polynomials.
In the case of Pan/Lichtblau's algorithm, the nature of the pairs being computed forces  us to relax the restrictions on S-polynomials, and limits the scope of the criteria.
In both cases, we prove that the algorithms are correct and compute both a signature-Gröbner basis of the ideal, and a basis of the signatures of its syzygies.

We have implemented both algorithms in the computer algebra system \textsf{Magma}~\cite{Magma}, with additional optimizations and criteria, and observe that the relaxed restrictions in Pan/Lichtblau tend to lead to the computation of more pairs.
We also compare the time taken for computing the signature Gröbner basis and using it to recover information on the module, with \textsf{Magma} implementations of \emph{ad-hoc} functions for that purpose, and show that using signatures allows for a significant speed-up of those operations.

\vspace{-0.2cm}
\section{Notations and Preliminaries}
\label{sec:notat-conv}

\subsection{Conventions and notations}
\label{sec:notations}

Let $\NN$ be the set of all non-negative integers.
Let $R$ be a principal ideal domain (PID) that has a unit element and is commutative.
We assume that $R$ is effective in the sense that one can perform all the arithmetic operations in $R$, obtain the $\gcd$ of elements and compute Bézout coefficients.
A typical example of such a ring is the ring of integers $\ZZ$, with Euclid's algorithm and its extended version.


%
Let $A = R[x_{1}, \dots,x_{\nvars}]$ be the polynomial ring in $\nvars$ indeterminates $x_{1},\dots,x_{\nvars}$ over $R$.
A monomial in $A$ is an element of the form $
x_{1}^{i_{1}} \dots x_{\nvars}^{i_{\nvars}}$ where $(i_{1},\dots,i_{\nvars}) \in \NN^{\nvars}$.
A term in $A$ is $a\mu$, where $a \in R \setminus \{0\}$ and $\mu$ is a monomial.
The set of terms (resp. monomials) of $A$ is denoted by $\Ter(A)$ (resp. $\Mon(A)$).
In the rest of the paper, unless stated otherwise, we shall use letters $a,b,\dots$ for coefficients in $R$ and $\mu, \nu, \dots$ for monomials in $\Mon(A)$.

A monomial order is an order on $\Mon(A)$ which is compatible with multiplication and well-founded.
In the rest of the paper, we assume that $\AA$ is endowed with an implicit monomial order $<$, and we define as usual the leading monomial $\LM$, the leading term $\LT$ and the leading coefficient $\LC$ of a given polynomial.
By convention, we set $\lm(0)=\lt(0)=\lc(0)=0$.

Given a pair of polynomials $(f,g)$, we denote $\lcmlm(f,g)$ (resp. $\lcmlt(f,g)$) the least common multiple of the leading monomials (resp. leading terms) of $f$ and $g$.

  
Given a set of polynomials $f_{1},\dots,f_{\npols}$ in $A$, we 
consider the free module $\mathbf{M} = A^{\npols}$ with basis $\bfe_{1},\dots,\bfe_{\npols}$.
For $\alpha \in \mathbf{M}$ with
$\alpha = (\alpha_{1},\dots,\alpha_{\npols})$, we define $\bar{\alpha} = \sum \alpha_{i}f_{i}$.
We define the module
\begin{equation}
  \label{eq:1}
  \II = \left\{ \sppair{\alpha} : \alpha \in \mathbf{M} \right\} \subset A^{\npols+1}.
\end{equation}
The module $\II$ is isomorphic to $\mathbf{M}$, and in particular it is free with basis $(\bfe_{1},f_{1}), \dots, (\bfe_{\npols},f_{\npols})$.
The image of the projection of $\II$ onto the last coordinate is the ideal, $\langle f_{1},\dots,f_{\npols}\rangle$.
A \emph{syzygy} of $\II$ is an element $\bfz = \sppair{\alpha} \in \II$ such that $\bar{\alpha}=0$.
The set of all syzygies of $\II$ is denoted by $\Syz(\II)$, it is a $A$-submodule of $\II$.

A monomial of $\mathbf{M}$ is an element of the form $\mu \bfe_{i}$, with $\mu \in \Mon(A)$
and $i \in \intint{1,\npols}$.
A term of $\mathbf{M}$ is an element of the form $c \mathbf{m}$ where $c \in R$ and $\mathbf{m}$ is a monomial of $\mathbf{M}$.
As before, the set of terms (resp. monomials) of $\mathbf{M}$ is denoted by $\Ter(\mathbf{M})$ (resp. $\Mon(\mathbf{M})$).
Given module terms $\mathbf{m}$ and $\mathbf{m}'$, we say that $\mathbf{m}$ is divisible by $\mathbf{m}'$ if there exists a term $t \in \Ter(A)$ such that $\mathbf{m} = t \mathbf{m}'$.

A monomial ordering on $\mathbf{M}$ is an ordering $\prec$ on $\Mon(\mathbf{M})$ such that for any $\mathbf{m}, \mathbf{n} \in \Mon(\mathbf{M})$ and $\mu,\nu \in \Mon(A)$,
\begin{enumerate}
  \item if $\mathbf{m} \prec \mathbf{n}$, then $\mu \mathbf{m} \prec \mu \mathbf{n}$;
  \item if $\mu < \nu$, then $\mu \mathbf{m} \prec \nu \mathbf{m}$.
\end{enumerate}

Examples of orderings on $\mathbf{M}$ are the \emph{position over term} (or PoT) ordering,
defined as $\mu \bfe_{i} \prec_{\textup{PoT}} \nu \bfe_{j}$ if $i < j$, or $i=j$ and $\mu < \nu$,
and the \emph{term over position} (or ToP) ordering,
defined as $\mu \bfe_{i} \prec_{\textup{ToP}} \nu \bfe_{j}$ if $\mu < \nu$,
or $\mu = \nu$ and $i < j$.

As in the case of polynomials, a monomial ordering on $\mathbf{M}$ can be extended into a partial term ordering.
Let $\bfs = a\mu\bfe_{i}$ and $\bft = b\nu\bfe_{j} \in \Ter(\mathbf{M})$, we write $\bfs \simeq \bft$ if $\bfs$ and $\bft$ are incomparable, that is, if $\mu=\nu$ and $i=j$. 
Equality $\bfs=\bft$ holds if 
$\bfs \simeq \bft$ and $a=b$.
We say that $\bfs \preceq \bft$ if $\bfs \prec \bft$ or $\bfs \simeq \bft$, and similarly, $\bfs \precneq \bft$ implies that $\bfs \nsimeq \bft$.
It is harmless because $\simeq$ is an equivalence relation and $\prec$ is a total order on the quotient, so, for example, if $\bfs \simeq \bft$ and $\bfs \prec \bfu$, then $\bft \prec \bfu$.

Given an element $\bfp = \sppair{\alpha} \in \II$, we define the leading term $\lt$, leading monomial $\lm$ and leading coefficient $\lc$ of $\bfp$ to be those of $\bar{\alpha}$.
The \emph{signature} of $\bfp$ is the leading term of the module element $\alpha$ for the module monomial ordering $\prec$, \emph{i.e.}, the largest module term appearing in $\alpha$, and it is denoted as $\sig(\bfp)$.

Readers familiar with signature-based algorithms over fields should note that in our setting, the signature of an element $\bfp$ is a \emph{term} $a\mu\bfe_{i}$, including a coefficient.
This will incur a minor overhead in adding and multiplying coefficients, but signature operations still have a negligible cost compared to operations on the polynomials.






\subsection{Signature Gröbner bases}
\label{sec:regul-s-polyn}

In this section, we introduce generalizations to rings of constructions used in signature Gröbner bases over fields.
These constructions extend those introduced in~\cite{FV2018}.

The key idea, as in the case of fields, is that for each element $\bff = \sppair{\alpha}$, we need only keep track of $\sig(\bff)=\lt(\alpha)$ and $\bar{\alpha}$, instead of the full module representation $\alpha$.
For that purpose, we restrict to operations which do not cancel the signatures.

\begin{definition}
  Let $\bff, \bfg \in \II$. 
  The sum $\bff + \bfg$ is called 
  \emph{regular} if $\sig(\bff) \nsimeq \sig(\bfg)$, and
  \emph{singular} if $\sig(\bff) = -\sig(\bfg)$.
\end{definition}
The nature of the operation yields information about the signature of the result, as follows.
\begin{proposition}
  Let $\bff=\sppair{\alpha}$ and $\bfg=\sppair{\beta} \in \II$, let $\bfh = \sppair{\gamma} = \bff + \bfg$.
  Then,
  \begin{itemize}
    \item $\bff + \bfg$ is a regular addition iff
    $\sig(\bfh) = \max(\sig(\bff),\sig(\bfg))$;
    \item $\bff + \bfg$ is a 
    non-singular addition iff $\sig(\bfh) = \sig(\bff) + \sig(\bfg) \simeq \sig(\bff) \simeq \sig(\bfg)$;
    \item $\bff + \bfg$ is a singular addition iff $\sig(\bfh) \precneq \sig(\bff) \simeq \sig(\bfg)$.
  \end{itemize}
\end{proposition}

The proof of the proposition is straightforward.
Note that in a singular addition, the signature of the result cannot be computed from the signatures of the summands. 
This phenomenon is called a \emph{signature drop}.
In order to avoid it and be able to keep track of the signatures, the algorithms must disallow singular operations.




Signature Gröbner bases, as in the case of fields, are characterized by the fact that all elements of the ideal are s-reducible, that is, reducible without increasing the signature.
In the case of rings, different notions of reduction exist, namely weak and strong reductions, as well as modular reductions by the coefficients.
In this paper, we only consider strong reductions
, which require that the leading coefficient of the reducer divides that of the reducee.
Those reductions allow to define strong Gröbner bases.
In the rest of the paper, we shall omit the ``strong'' qualificative.



\begin{definition}
\label{weak-sig-reduction}
Let $\mathcal{G} 
\subset \II$
and $\bff,\bfh 
\in \II$.
  We say that $\bff$ \emph{(strongly) \redpol-reduces} to $\bfh$ modulo $\GG$ if there exists $\bfg_{i} \in \GG$ and $t_{i} \in \Ter(A)$ such that
  \begin{enumerate}
    \item $\lt(\bff) = t_{i} \lt(\bfg_{i})$
    \item $\bfh = \bff - t_{i} \bfg_{i}$ 
    \item $t_{i} \sig(\bfg_{i}) \preceq \sig(\bff)$
  \end{enumerate}
    If the signature inequality is strict, $t_{i}\sig(\bfg_{i}) \precneq \sig(\bff)$, it is called a \emph{regular} \redpol-reduction, and if $t_{i}\sig(\bfg_{i}) = \sig(\bff)$, it is called a \emph{singular} \redpol-reduction. 
    It is called a s-reduction to zero if the polynomial part of $\bfh$ is zero.
  By abuse of language, we extend these definitions to sequences of reductions. 
\end{definition}

  If $\bff$ \redpol-reduces to $\bfh$ modulo $\GG$, then $\sig(\bfh) \preceq \sig(\bfg)$, with equality if and only if the reduction is regular and strict inequality if and only if the reduction is singular.
  Note that an s-reduction might be neither regular nor singular, in which case $\sig(\bfh) \simeq \sig(\bff)$.

We then recall the definition of (strong) signature Gröbner bases.\footnote{In \cite{Gao-2015-new-framework-for}, a signature GB is called a strong GB. We use Sig-GB here to avoid conflict with the existing notion of strong GB over rings.}

\begin{definition}
  Let $\GG \subset \II$ and $\bfT \in \Ter(\mathbf{M})$.
  $\GG$ is called a (strong) signature Gröbner basis (or Sig-GB for short) up to signature $\bfT$ if 
    every $\bff \in \II$ with $\sig(\bff) \precneq \bfT$ is $\redpol$-reducible modulo $\GG$. 
  $\GG$ is called a signature Gröbner basis if it is a signature Gröbner basis up to $\bfT$ for all $\bfT$.
\end{definition}

The original motivation for the use of signatures is to maintain a list of signatures of known syzygies, and use it to predict reductions to zero.
Additionally, the last coordinates of elements of a Sig-GB form a GB in the classical sense.
The proof of that fact~\citep[Lem.~4.6]{eder:2017:survey} can be directly extended to rings.
So signature-based algorithms allow to compute classical Gröbner bases in a more efficient way.

This use of syzygies applies to our case as well, and requires to define reductions by signatures of syzygies.

\begin{definition}
    Let $\GG_{z} \subset \Syz(\II) \subset \II$ and let $\bff \in \II$, with $\sig(\bff) = a\mu\bfe_{i}$, for $a \in R$ and $\mu \in \Mon(A)$.
    We say that $\bff$ is \emph{$\redsig$-reducible} modulo $\GG_{z}$ if there exists $\bfz \in \GG_{z}$ such that $\sig(\bfz)$ divides $\sig(\bff)$.

  Let $\bfT \in \Mon(\mathbf{M})$ , we say that $\GG_{z}$ is a Sig-basis of syzygies (resp. basis up to $\bfT$) if any syzygy of $\II$ (resp. syzygy with signature $\precneq \bfT$) is sig-reducible by $\GG_{z}$.
\end{definition}

\begin{proposition}
  \label{prop:syz-basis}
  Let $\GG$ be a Sig-GB up to signature $\bfT$, $\GG_{z} \subset \Syz(\II)$ 
  and $\bff \in \II$ with $\sig(\bff) \preceq \bfT$.
  If $\bff$ is sig-reducible modulo $\GG_{z}$, then $\bff$ regular s-reduces to 0 modulo $\GG$.
\end{proposition}
\begin{proof}
  Let $\bfz \in \GG_{z}$ be such that there exists $t \in \Ter(A)$ with $t \sig(\bfz) = \sig(\bff)$.
  Let $\bfg = \bff - t\bfz$, it has signature $\precneq \bfT$ so, by assumption on $\GG$, it s-reduces to 0 modulo $\GG$.
  Furthermore, since $\bfz$ is a syzygy, by definition its polynomial part is $0$, so the polynomial part of $\bfg$ is equal to that of $\bff$.
  Thus the s-reduction to 0 of $\bfg$ is a s-reduction to 0 of $\bff$, and from the observation on the signatures, it is regular.
\end{proof}

This kind of ``least criminal'' argument, relying on the notion of Sig-GB \emph{up to} some signature, will be a recurring pattern in the proofs.

In the classical case, without signatures, it is sometimes convenient to consider expanded sequences of reductions, leading to the notion of standard representation. 
With signatures, it turns out that a natural generalization of that notion encompasses both s-reductions and sig-reductions.

\begin{definition}
\label{standard-sig-representation}
Let $\mathcal{G} = \{\bfg_{1},\dots,\bfg_{r}\} \subset \II$, $\GG_{z}=\{\bfz_{1},\dots,\bfz_{s}\} \subset \Syz(\II)$ and $\bfh \in \II$. 
Let  $t^{(1)}_{u} \in \Ter(A)$, $i_{u} \in \intint{1,r}$, where $u \in \intint{1,k}$ and $k \in \mathbb{N}$, $t_{v}^{(2)} \in \Ter(A)$, $j_{v} \in \intint{1,s}$, where $v \in \intint{1,l}$ and $l \in \mathbb{N}$, be such that the equality 
\begin{equation}
  \label{eq:19}
  \textstyle\bfh = \sum_{u=1}^{k} t^{(1)}_{u}\bfg_{i_{u}} + \sum_{v=1}^{l} t_{v}^{(2)}\bfz_{j_{v}}
\end{equation}
holds in $\II$, with
\begin{enumerate}
  \item\label{item:sigrep:lt} $\lt(t^{(1)}_{1}\bfg_{i_{1}}) > \lt(t^{(1)}_{2}\bfg_{i_{2}}) \geq \lt(t^{(1)}_{3}\bfg_{i_{3}}) \geq \dots \geq \lt(t^{(1)}_{k}\bfg_{i_{k}})$;
  \item\label{item:sigrep:sigpol} for all $u\in\intint{1,k}$, $\sig(t_{u}^{(1)}\bfg_{i_{u}}) \preceq \sig(\bfh)$;
  \item\label{item:sigrep:sig} $\sig(t_{1}^{(2)}\bfz_{j_{1}}) \succneq \sig(t_{2}^{(2)}\bfz_{j_{2}}) \succeq \sig(t_{3}^{(2)}\bfz_{j_{3}}) {\succeq} \dots {\succeq} \sig(t_{l}^{(2)}\bfz_{j_{l}})$;
  \item\label{item:sigrep:sigsyz} for all $v \in \intint{1,l}$, $\sig(t_{v}^{(2)}\bfz_{j_{v}}) \preceq \sig(\bfh)$.
\end{enumerate}
If such a decomposition exists, we say that \eqref{eq:19} is a \emph{standard Sig-representation} of $\bfh$ with respect to $(\GG,\GG_{z})$.
\end{definition}

\begin{proposition}
  Let $\GG \subset \II$, $\GG_{z} \subset \Syz(\II)$ such that every element of $\II$ admits a standard Sig-representation by $(\GG,\GG_{z})$.
  Then $\GG$ is a Sig-GB and $\GG_{z}$ is a Sig-basis of syzygies.
\end{proposition}
\begin{proof}
  Let $\bfh \in \II$. By assumption it admits a standard Sig-representation as in~\eqref{eq:19}.
  If $\bfh$ is not a syzygy, then by property~\ref{item:sigrep:lt} 
  on the leading terms, $\lt(\bfh) = t_{1}^{(1)}\lt(\bfg_{i_{1}})$, and by property~\ref{item:sigrep:sigpol} on the signatures, this is an s-reduction of $\bfh$.
  If $\bfh$ is a syzygy, then again by the property on the leading terms, $k=0$, and properties~\ref{item:sigrep:sig} 
  and \ref{item:sigrep:sigsyz} on the signatures imply that $\bfh$ is sig-reducible by $\bfz_{j_{1}}$.
\end{proof}

We recall how S-polynomials are defined with signatures.
Recall that the notations $\lt$, $\lm$, $\lc$ stand for the leading term, monomial, coefficient of the polynomial part of module elements, and that the notation $\sig$ is used for the module leading term.
First, we give some definitions associated with pairs of module elements.
\begin{definition}
  Let $\bff, \bfg \in \II$, $t_{\bff} = \frac{\lcmlt(\bff,\bfg)}{\lt(\bff)}$,
  $t_{\bfg} = \frac{\lcmlt(\bff,\bfg)}{\lt(\bfg)}$.
  
    The term degree of the pair $(\bff,\bfg)$ is $\ltS(\bff,\bfg) = \lcmlt(\bff,\bfg) = t_{\bff}\lt(\bff) = t_{\bfg}\lt(\bfg)$.
    The monomial degree $\lmS(\bff,\bfg)$ of the pair $(\bff,\bfg)$ is the monomial part of the term degree.

      The pair $(\bff,\bfg)$ is called regular if $t_{\bff}\sig(\bff) \nsimeq t_{\bfg}\sig(\bfg)$ and it is called singular if $t_{\bff}\sig(\bff) + t_{\bfg}\sig(\bfg) = 0$.
      The signature of the pair $(\bff,\bfg)$ is
      $\sigS(\bff,\bfg) = \max(t_{\bff}\sig(\bff), -t_{\bfg}\sig(\bfg))$.
\end{definition}

\begin{definition}
  \label{strong-S-poly-sig}
  Let $\bff
  $ and $\bfg
  \in \II$.
  The S-polynomial of $\bff$ and $\bfg$ is
  \begin{equation}
    \label{eq:73}
    \SPol(\bff,\bfg) = 
    \frac{\lcmlt(\bff,\bfg)}{\lt(\bff)}\bff
    - \frac{\lcmlt(\bff,\bfg)}{\lt(\bfg)}\bfg.
  \end{equation}
  \end{definition}
Let $\bfh= \SPol(\bff,\bfg)$.
Then $\sig(\bfh) \preceq \sigS(\bff,\bfg)$, with equality iff the pair $(\bff,\bfg)$ is regular, and strict inequality iff it is singular.

\begin{example}
  Let $\bff_{1} = (\bfe_{1},4xy + 1)$ and $\bff_{2}=(\bfe_{2},6x^{2}+1)$.
  Then $\bfh=\SPol(\bff_{1},\bff_{2}) = 3x\bff_{1} - 2y\bff_{2} = (3x\bfe_{1}-2y\bfe_{2}, 3x-2y)$.
  With the PoT ordering with $\bfe_{1}\succ\bfe_{2}$, $\sig(\bfh)=3x\bfe_{1} = \sigS(\bff_{1},\bff_{2})$, and the pair $(\bff_{1},\bff_{2})$ is regular.
\end{example}

In order to ensure that elements are strongly $\redpol$-reducible modulo $\GG$, we need to compute G-polynomials%
\footnote{Terminology and notations vary:
  this construction is called T-polynomial in~\cite{Moller:1988:grobnerrings2},
  $SL$ in \cite{Pan:Dbases},
  ``CP2 type (a)'' pairs in~\cite{Kandri-Rody-Kapur},
  S-polynomial of type 1 in~\cite{Lichtblau},
  G-polynomial in~\cite{Becker:1993:grobnerbasistext}
  and
  GCD-polynomial in~\cite{Eder:2017:EuclideanRings,EderPopsecu:Standardbases:2018}.}.
The G-polynomial of $f_{1}$ and $f_{2}$ is a polynomial $f$ such that any linear combination of $f_{1}$ and $f_{2}$ not cancelling the leading terms is reducible by $f$.
It is defined by using Bézout relations to make the leading coefficient as small as possible.
\begin{definition}
  Let $\bff$ and $\bfg \in \II$.
  Let $u,v$ be Bézout coefficients of $\lc(\bff)$ and $\lc(\bfg)$, that is, $u\lc(\bff) + v\lc(\bfg) = \gcd(\lc(\bff),\lc(\bfg))$.
  The G-polynomial of $\bff$ and $\bfg$ associated to $(u,v)$ is defined as
  \begin{equation}
    \label{eq:15}
    \GPol_{u,v}(\bff,\bfg) = u \frac{\lcmlm(\bff,\bfg)}{\lm(\bff)} \bff
    + v \frac{\lcmlm(\bff,\bfg)}{\lm(\bfg)} \bfg.
  \end{equation}
\end{definition}

\begin{figure*}[t]
  \removelatexerror
  \begin{minipage}{\columnwidth}
    \small
    \input{algo_KRK}
    \vspace{-0.2cm}
  \end{minipage}%
  \hspace{\columnsep}%
  \begin{minipage}{\columnwidth}
    \small
    \input{algo_Lichtblau}
    \vspace{-0.2cm}
  \end{minipage}
\end{figure*}


The coefficients $u$ and $v$ are not uniquely determined, and we can use this fact to ensure that G-polynomials \emph{never} represent a singular operation.
\begin{proposition}
  \label{prop:G-pol-nonsing}
  Let $\bff$ and $\bfg \in \II$.
  Then there exists $u,v$ such that $\sig(\GPol_{u,v}(\bff,\bfg)) \simeq \sig(\bff,\bfg)$.
\end{proposition}
\begin{proof}
  If the pair is regular, there is nothing to prove, and any pair of Bézout coefficients works.
  Otherwise, let $a = \lc(\bff)$, $b = \lc(\bfg)$, $c = \lc(\sig(\bff))$, $d = \lc(\sig(\bfg))$, and $g = \gcd(a,b)$.
  We want to prove that there exists $u,v$ such that $au+bv = g$ and $cu+dv \neq 0$.
  If $ad-bc = 0$, $a(cu + dv) = c (au + bv) \neq 0$, so again any pair of Bézout coefficients works.
  Otherwise, assume that $au+bv=0$, $cu+dv =0$, and consider the pair $u' = u + b$, $v' = v-a$.
  Then $au'+bv' = g$ and $cu'+dv' =bc - ad \neq 0$.
\end{proof}

In practice, we shall always consider such a pair of Bézout coefficients, and call the corresponding G-polynomial \emph{the} G-polynomial of $\bff$ and $\bfg$, denoted by $\GPol(\bff,\bfg)$.
Note that $\sig(\GPol(\bff,\bfg)) \simeq \sigS(\bff,\bfg)$ and $\lm(\GPol(\bff,\bfg)) = \lmS(\bff,\bfg)$.

\begin{definition}
  Let $\GG \subset \II$.
  We say that $\GG$ is \emph{complete} if every G-polynomial of elements of $\GG$ is s-reducible modulo $\GG$.
\end{definition}
It is always possible to make $\GG$ complete by adding G-polynomials to $\GG$ until the property holds.
We use a similar process, but on the signatures, to create syzygies with small signature coefficients.
\begin{definition}
  Let $\bfz_{1}, \bfz_{2} \in \Syz(\II)$ with respective signatures $a_{k}\mu_{k}\bfe_{i}$, $k = 1,2$, sharing the same index $i$.
  Let $d = \gcd(a_{1},a_{2})$, and let $u_{1}, u_{2}$ be Bézout coefficients.
  Let $\mu = \lcm(\mu_{1},\mu_{2})$.
  The $\sigG$-combination of $\bfz_{1}$ and $\bfz_{2}$ is defined as
  \begin{equation}
    \label{eq:18}
    \sigGcomb(\bfz_{1},\bfz_{2}) = u_{1} \frac{\mu}{\mu_{1}} \bfz_{1} + u_{2} \frac{\mu}{\mu_{2}} \bfz_{2},
  \end{equation}
  and its signature is $d\mu\bfe_{i}$.
%
  Let $\GG_{z} \subset \Syz(\II)$, we say that $\GG_{z}$ is $\sigG$-complete if any $\sigG$-combination 
  of elements of $\GG_{z}$ is sig-red. by $\GG_{z}$.
\end{definition}

We conclude this section with a few definitions which will give useful criteria  to prove the correctness and for detecting useless syzygies in  Algorithm~\ref{algo:KRK-like} given below.
These constructions are adapted from the ones defined for the GVW algorithm over fields \cite{Gao:2010:GVW}.
The first definition is that of super reducible elements\footnote{The notion corresponds to that of \emph{sig-redundant} elements in~\cite{eder:2011:signature}. In the paper \cite{FV2018}, it was called 1-singular reducible.}.
\begin{definition}
  Let $\GG \subset \II$, and $\bff \in \II$.
  $\bff$ is \emph{super reducible} by $\GG$ if there exists $\bfg \in \GG$ and $t \in \Ter(A)$ such that $\sig(\bff) = t \sig(\bfg)$ and $\lm(t\bfg) = \lm(\bff)$.
\end{definition}

Note that unlike in the case of fields, we do not require that a super reduction is a reduction: the leading monomials match, but the leading coefficients need not match.
However, an element which is super reducible is necessarily $\redpol$-reducible by some G-polynomial of the basis elements.

\begin{proposition}
  \label{prop:super-red-implies-red}
  Let $\GG \subset \II$ be complete, and 
  let $\bff \in \II$.
  If $\bff$ is super reducible by $\GG$, then $\bff$ is $\redpol$-reducible by $\GG$.
\end{proposition}
\begin{proof}
  Assume for a contradiction that $\bff$ is super reducible by $\GG$, not $\redpol$-reducible, and has minimal signature for this property.
  Let $\bfg_{1}$ be such that there exists $t_{1} \in \Ter(A)$ with $\sig(\bff) = t_{1} \sig(\bfg_{1})$ and $\lm(t_{1}\bfg_{1}) = \lm(\bff)$.
  In particular, $\sig(\bff-t_{1}\bfg_{1}) \precneq \sig(\bff)$.
  If $\lt(t_{1}\bfg_{1}) = \lt(\bff)$, then $\bfg_{1}$ is a $\redpol$-reducer of $\bff$.
  
  Otherwise, $\lm(\bff - t_{1}\bfg_{1}) = \lm(\bff)$.
  Since $\sig(\bff-t_{1}\bfg) \precneq \sig(\bff)$, by minimality of $\sig(\bff)$, $\bff - t_{1}\bfg$ is $\redpol$-reducible modulo $\GG$.
  Let $\bfg_{2}$ be such a reducer, with $\lt(\bff-t_{1}\bfg_{1}) = t_{2}\bfg_{2}$, and so $\lt(\bff) = t_{1}\lt(\bfg_{1}) + t_{2}\lt(\bfg_{2})$.
  The signature satisfies $t_{2}\sig(\bfg_{2}) \preceq \sig(\bff - t_{1}\bfg_{1}) \precneq \sig(t_{1}\bfg_{1})$.
  Let $\bfg_{3} = \GPol(\bfg_{1},\bfg_{2})$, by definition of the G-pol there exists $t_{3} \in \Ter(A)$ such that $t_{3}\lt(\bfg_{3}) = \lt(\bff)$, and $t_{3}\sig(\bfg_{3}) \simeq t_{1}\sig(\bfg_{1}) \simeq \sig(\bff)$.
  By hypothesis $\bfg_{3}$ is $\redpol$-reducible by $\GG$, and a $\redpol$-reducer of $\bfg_{3}$ is a $\redpol$-reducer of $\bff$.
\end{proof}

The last definition is that of the covered property, which can replace reduction to 0 as a correctness criterion for signature algorithms.
The advantage is twofold: it makes the proof of correctness more straightforward (as can be seen on the proofs of Th.~\ref{prop:correctness-KRK}, which uses the cover criterion, and of Th.~\ref{prop:correctness-Lichtblau}, which does not); and since it is a weaker condition than reduction to 0, it allows to discard more pairs in the course of the algorithm.
\begin{definition}
  Let $(\bff_{1},\bff_{2}) \in \II^{2}$ be a pair.
  Let $\GG \subset \II$ and $\GG_{z} \subset \Syz(\II)$.
  The pair $(\bff_{1},\bff_{2})$ is \emph{covered by $(\GG,\GG_{z})$} if there exists $\bfg \in \GG$,
  $\bfz \in \GG_{z}$, $t,t^{(z)} \in \Ter(A)$ such that
  \begin{itemize}
    \item if $t \neq 0$, $\sig(\bff,\bfg) \simeq t\sig(\bfg)$;
    \item if $t^{(z)}\neq 0$, $\sig(\bff,\bfg) \simeq t^{(z)}\sig(\bfz)$;
    \item $\sig(\bff,\bfg) = t\sig(\bfg) + t^{(z)}\sig(\bfz)$;
    \item $\lm(t\bfg) < \lmS(\bff,\bfg)$.
  \end{itemize}
\end{definition}

This cover criterion looks more complicated to implement than in the case of fields, due to the need to consider linear combinations.
However, one can use $\sigG$-combinations of elements of $\GG$ and elements of $\GG_{z}$ to compute elements with signature as small as possible and same leading monomial, and reduce the cover test to a single divisibility test.

\section{Adding signatures to Kandri-Rody and Kapur's algorithm}
\label{sec:moellers-algorithm}

\subsection{Description of the algorithm}
\label{sec:descr-algor}

The first algorithm which we present in this paper is a signature-enabled version of Kandry-Rody and Kapur's algorithm. The algorithm works similarly to Buchberger's algorithm, but adds both S- and G-polynomials to the basis.
The signature variant follows the construction of the GVW algorithm~\cite{Gao:2010:GVW, Gao-2015-new-framework-for}.
An example run of the algorithm, as well as the one in the next section, is available online\footnote{\url{https://gitlab.com/thibaut.verron/signature-groebner-rings/-/raw/public/ISSAC21_example.pdf}}.

The correctness of the algorithm is stated by the following theorem (proved in Section~\ref{sec:proof}), and adapted from \cite[Thm.~2.4]{Gao-2015-new-framework-for}.
\begin{theorem}
  \label{prop:correctness-KRK}
  Let $\GG \subset \II$ complete , $\GG_{z} \subset \Syz(\II)$ $\sigG$-complete, such that for all signatures $\bfT$, there exists some $\bfg \in \GG \cup \GG_{z}$ such that $\sig(\bfg)$ divides $\bfT$.
  Assume that every regular pair of elements of $\GG$ is covered by $(\GG,\GG_{z})$.
  Then,
    $\GG$ is a Sig-Gröbner basis and 
    $\GG_{z}$ is a Sig-basis of syzygies.
\end{theorem}

The algorithm ensures that all the assumptions 
hold:
\begin{itemize}
  \item $\GG$ is complete and $\GG_{z}$ is $\sigG$-complete because G-polynomials are added to the queue of pairs to be reduced for addition into $\GG$, and $\sigG$-combinations to $\GG_{z}$;
  \item there exists, for all $\bfT$, a $\bfg$ with $\sig(\bfg)$ dividing $\sig(\bfT)$, because we process all elements $(\bfe_{i},f_{i})$, thus ensuring that there is an element with signature $\bfe_{i}$ in either $\GG$ or $\GG_{z}$ for all $i$;
  \item every regular pair is covered by $(\GG,\GG_{z})$ because, for each regular pair, we compute the corresponding S-polynomial, reduce it and add the result to the basis, thus creating a covering element for the pair.
\end{itemize}

The resulting algorithm is described in Algorithm~\ref{algo:KRK-like}.
Note that for each element $\bff=\sppair{\alpha} \in \II$, we only keep track of  $\sig(\bff) = \lt(\alpha)$ and $\bar{\alpha}$.
The routines $\mathsf{SigReduce}$ and $\mathsf{RegularReduce}$ compute the sig-reduction of a signature by a basis of syzygies (the result being either $0$ or the signature itself), and the regular reduction of an element of $\II$ by a Sig-GB, respectively.

A technical point is that the theorem allows to eliminate S-polynomials obtained from a pair which is covered, but not necessarily G-polynomials.
This requires to keep track of how each element was computed.
In the pseudo-code algorithm, we do it by keeping for each new element its so-called \textsf{type}, which can take three values: \textsf{N}, indicating a polynomial from the input; $\mathsf{S}(i,j)$, indicating the S-polynomial of $\bfg_{i}$ and $\bfg_{j}$; and $\mathsf{G}(i,j)$, indicating the G-polynomial of $\bfg_{i}$ and $\bfg_{j}$.

On top of that, we add some tests to eliminate some G-polynomials.
Firstly, if $\lc(\bfg_{i})$ divides $\lc(\bfg_{j})$, then one can choose the Bézout coefficients such that $\GPol(\bfg_{i},\bfg_{j})$ is a multiple of $\bfg_{i}$, and thus it is automatically s-reducible modulo $\GG$.

Secondly, thanks to Proposition~\ref{prop:syz-basis}, we know that we can immediately disregard any element whose signature is divisible by that of a syzygy.
This partially extends the cover criterion to G-polys.

Thirdly, note that we cannot use Proposition~\ref{prop:super-red-implies-red} to eliminate G-polynomials which would be super reducible: indeed, that proposition requires that $\GG$ be complete, and the G-polynomial being processed might be necessary for that.
Furthermore, the G-polynomial $\GPol(\bfg_{i},\bfg_{j})$ is always super reducible by at least one of $\bfg_{i}, \bfg_{j}$.

However, if $\bff \in \II$ is both super reducible and s-reducible by $\bfg \in \GG$, that is, if there exists $t \in \Ter(A)$ such that $t \sig(\bfg) = \sig(\bff)$ and $t \lt(\bfg) = \lt(\bff)$, then the proof of Proposition~\ref{prop:super-red-implies-red} shows that $\bff$ is s-reducible by $\GG$, without any hypothesis of completeness.
Thus such an element can be immediately discarded.

For some signature orderings, it is also possible to predict in advance the signature of some syzygies, with the F5 criterion.
This criterion can be implemented in our setting exactly as in the case of fields, for instance by adding signatures to $\GG_{z}$, and we do not detail it here.

\subsection{Proof}
\label{sec:proof}
The proof of Theorem~\ref{prop:correctness-KRK} is adapted from the proof of \cite[Thm.~2.4]{Gao-2015-new-framework-for}.


\begin{proof}[Proof of Th.~\ref{prop:correctness-KRK}]
  We prove the implication by contradiction.
  Assume that there exists $\bff \in \II$ such that $\bff$ is not s-reducible modulo $\GG$ and $\bff$ is not sig-reducible modulo $\GG_{z}$, and pick $\bff$ with minimal signature $\bfT$ for this property.
  Let $\bfg_{1} \in \GG$, $\bfz_{1} \in \GG_{z}$, $t,t_{z_{1}} \in \Ter(A)$ such that $\bfT = t\sig(\bfg_{1}) + t_{z_{1}}\sig(\bfz_{1})$ and such that $\lm(t\bfg_{1})$ is minimal for that property.
  By hypothesis, such a decomposition exists (with either $t_{1}=0$ or $t_{z_{1}}=0$).

  First, we prove that $t\bfg_{1}$ is not regular $\redpol$-reducible modulo $\GG$.
  Indeed, if it were, let $\bfg_{2}$ be a regular reducer of $t\bfg_{1}$.
  Consider the pair $(\bfg_{1},\bfg_{2})$, let $\mu = \lcmlm(\bfg_{1},\bfg_{2})$ and let $\sigma = \sigS(\bfg_{1},\bfg_{2})$.
  By properties of the lcm, there exist some terms $t_{1},t_{2}$ such that $\mu= \lm(t_{1}\bfg_{1})$, $\sigma = t_{1} \sig(\bfg_{1})$, and $t_{1}t_{2} = t$.
  Furthermore, since $\bfg_{2}$ is a regular reducer of $t\bfg_{1}$, the pair $(\bfg_{1},\bfg_{2})$ is regular (\cite[Lemma.~2.3]{Gao-2015-new-framework-for}).
    
  By assumption, the pair is covered by $(\GG,\GG_{z})$,
  so there exists $\bfg_{3} \in \GG$, $\bfz_{2} \in \GG_{z}$, $t_{3},t_{z_{2}} \in \Ter(A)$,
  such that $\sigma = t_{3}\sig(\bfg_{3}) + t_{z_{2}}\sig(\bfz_{2})$,
  and $\lm(t_{3}\bfg_{3}) < \mu$.
  So all in all, $\bfT = t_{2}t_{3}\sig(\bfg_{3}) + t_{z_{1}}\sig(\bfz_{1}) + t_{z_{2}}\sig(\bfz_{2})$,
  and $\lm(t_{2}t_{3}\bfg_{3}) < \lm(t_{1}t_{2}\bfg_{1}) = \lm(t\bfg_{1})$.
  Let $\bfz_{3}$ be  $\sigGcomb(\bfz_{1},\bfz_{2})$,
  its signature divides the sum $t_{z_{1}}\sig(\bfz_{1}) + t_{z_{2}}\sig(\bfz_{2})$, and thus, the existence of  $(\bfg_{3},\bfz_{3})$ contradicts the minimality of $\lm(g_{1})$.
  
  So $t \bfg_{1}$ is not regular $\redpol$-reducible modulo $\GG$.
  Now we consider two distinct cases, depending on whether $\bff$ is a syzygy or not.

  If $\bff$ is not a syzygy, $\lm(\bff) \neq \lm(t\bfg_{1})$, because otherwise $\bff$ would be super reducible by $\bfg_{1}$ and thus, since $\GG$ is complete,  s-reducible by $\GG$.
  Let $\bff_{1} = \bff - t\bfg_{1} - t_{z_{1}}\bfz_{1}$, so $\sig(\bff_{1}) \precneq \bfT$, and $\lt(\bff_{1}) = \max(\lt(\bff), t\lt(\bfg_{1}))$.
  Since $\sig(\bff_{1}) \precneq \bfT$, by minimality of $\bff$, $\bff_{1}$ s-reduces to 0 modulo $\GG$. 
  But since $\lt(\bff_{1}) = \max(\lt(\bff), t\lt(\bfg_{1}))$, any $\redpol$-reduction of $\bff_{1}$ is a regular reduction of either $\bff$ or $t\bfg_{1}$, which is a contradiction.

  If $\bff$ is a syzygy, we proceed similarly, but now $\lm(\bff)=0$.
  So the fact that $\bff_{1}$ is s-reducible implies that $t\bfg_{1}$ must be regular reducible, which is impossible.
  So $t\bfg_{1}=0$ and $\bff$ is sig-reducible by $\bfz_{1}$.
\end{proof}

\section{Adding signatures to Pan/Lichtblau's algorithm}
\label{sec:PanL-algor}

\subsection{Description of the algorithm}
\label{sec:Descr-algor-1}

The second algorithm which we present is adapted from that of Lichtblau~\cite{Lichtblau}, which is itself adapted from that of Pan~\cite{Pan:Dbases}.
The main difference with the previous algorithm is that it tries to limit the growth of the length of the queue 
by adding at most one new polynomial for each pair, either an S- or a G-polynomial, by the following construction.

\begin{definition}
  Let $\bff,\bfg \in \II$.
  The SG-polynomial 
  of $\bff$ and $\bfg$ is
  \begin{equation}
    \label{eq:6}
    \SGPol(\bff_{1},\bff_{2}) =
    \begin{cases}
      \SPol(\bff,\bfg) & \text{if } \lc(\bff) \divides \lc(\bfg) \text{ or } \lc(\bfg) \divides \lc(\bff)\\
      \GPol(\bff,\bfg) & \text{otherwise}.
    \end{cases}
  \end{equation}
\end{definition}

The reason why this construction is sufficient is that if $\lc(\bff)$ and $\lc(\bfg)$ do not divide each other, then the S-polynomial of $\bff$ and $\bfg$ can be expressed in terms of the S-polynomials of $\bff$, $\bfg$ and $\bfh=\GPol(\bff,\bfg)$.
However, it forces us to also compute non-regular S-polynomials as long as they are non-singular: indeed, $\sigS(\bff,\bfg) \simeq \sig(\bfh)$, and $\sig(\bfh) \simeq \sigS(\bff,\bfh) \simeq \sigS(\bfg,\bfh)$ because $\lm(\bfh)$ is equal up to coefficient to the term degree of $\SPol(\bff,\bfg)$, $\SPol(\bff,\bfh)$ and $\SPol(\bfg,\bfh)$.
The element $\SPol(\bff,\bfg)$ is a linear combination of $\SPol(\bff,\bfh)$ and $\SPol(\bfg,\bfh)$, and this combination cannot be regular.

Compared to the previous algorithm, considering SG-polynomials seems to lead to to computing fewer S-polynomials and the exact same G-polynomials, since it is always useless to compute a G-polynomial when one of the leading coefficients divides the other.
However, allowing non-regular S-polynomials means computing more S-polynomials, and there is no straightforward theoretical comparison between the algorithms.
Allowing non-regular S-polynomials also means that we cannot \emph{a priori} use the cover criterion in our algorithm: the algorithm would not ensure that all regular pairs are covered, but rather, only those for which the SG-polynomial is actually an S-polynomial.
We can however eliminate S-polynomials which are super reducible, since Proposition~\ref{prop:super-red-implies-red} ensures that they s-reduce to 0.

The rest of the algorithm, including the processing of syzygies, is done in exactly the same way as in Algorithm~\ref{algo:KRK-like}.
%
%
The correctness of the algorithm is ensured by the following theorem, which will be proved in Section~\ref{sec:Proof-1}.
\begin{theorem}
  \label{prop:correctness-Lichtblau}
  Let $\GG 
  \subset \II$ complete and $\GG_{z} \subset \Syz(\II)$ $\sigG$-complete such that 
  \begin{itemize}
    \item $\forall\,i \in \intint{1,\npols}$, $(\bfe_{i},f_{i})$ has a standard Sig-rep. w.r.t. $(\GG,\GG_{z})$
    ;
    \item any non-singular SG-polynomial of elements of $\GG$ 
  has a standard Sig-representation w.r.t. $(\GG,\GG_{z})$.
\end{itemize}
Then $\GG$ is a Sig-Gröbner basis and $\GG_{z}$ is a Sig-basis of syzygies.
\end{theorem}

\vspace{-0.2cm}
\subsection{Proof}
\label{sec:Proof-1}

The proof of Theorem~\ref{prop:correctness-Lichtblau} is adapted from that of \cite{Pan:Dbases} and \cite {Lichtblau} with
the addition of signatures.
First, we prove a useful technical lemma.

\begin{lemma}
  \label{lemme:1}
   Let $\mathcal{G} = \{\bfg_{1},\dots,\bfg_{r}\} \subset \II$ and $\bfT \in \Ter(\mathbf{M})$ 
   such that
   \begin{itemize}
     \item for all $\bfg \in \II$ with $\sig(\bfg) \precneq \bfT$, $\bfg$
     $\redpol$-reduces to 0 modulo $\GG$;
    \item for all $\bfg_{i}, \bfg_{j} \in \GG$ such that $\SGPol(\bfg_{i},\bfg_{j})$ is non-singular and $\sigS(\bfg_{i},\bfg_{j}) \preceq \bfT$, $\SGPol(\bfg_{i},\bfg_{j})$ $\redpol$-reduces to $0$ modulo $\GG$.    
  \end{itemize}
  Let $\bfg_{i},\bfg_{j} \in \GG$  
  such that $\lc(\bfg_{j})$ divides $\lc(\bfg_{i})$.
  Then any (possibly singular) linear combination of $\bfg_{i}$ and $\bfg_{j}$ with signature at most $\bfT$ $\redpol$-reduces to $0$ modulo $\GG$.
\end{lemma}
\begin{proof}
  Let $\mm{i} = \lm(\bfg_{i})$,  $\mm{j} = \lm(\bfg_{j})$, and $\mm{i,j} = \lcmlm(\bfg_{i},\bfg_{j})$. Consider a linear combination $\bfh = t_{i}\bfg_{i} + t_{j} \bfg_{j}$.
  Without loss of generality, combining the reductions, we may assume that $t_{i}$ and $t_{j}$ are terms.
  If $\LM(t_{i}\bfg_{i}) \neq \LM(t_{j}\bfg_{j})$, say $\LM(t_{i}\bfg_{i}) > \LM(t_{j}\bfg_{j})$, there is nothing to prove, as $\bfh$ can be
  reduced by $\bfg_{i}$, then $\bfg_{j}$.

  So assume that $\LM(t_{i}\bfg_{i}) = \LM(t_{j}\bfg_{j})$.
  Note that $\mm{i,j}$ divides the common multiple
  $\LM(t_{i}\bfg_{i}) = \LM(t_{j}\bfg_{j})$, say, $\LM(t_{i}\bfg_{i}) = t' \mm{i,j}$.
  By assumption on the leading coefficients, there exists $c \in R$ such that
  $\LT(t_{i}\bfg_{i})
  = t_{i}t' \frac{\mm{i,j}}{\mm{i}} \LT(\bfg_{i})
  = t_{i}t' c \frac{\mm{i,j}}{\mm{j}} \LT(\bfg_{j})$.

  If $\lm(\bfh) = \lm(t_{i}g_{i})$, the leading term of $\bfh$ is
  \(
    \LT(\bfh) = \LT(t_{i}\bfg_{i}) + \LT(t_{j}\bfg_{j})
    = (t_{i}c + t_{j}) t' \frac{\mm{i,j}}{\mm{j}} \LT(\bfg_{j})
  \)
  and $\bfh$ is reducible by $\bfg_{j}$.
  This reduction is a $\redpol$-reduction by construction.

  The remaining case is the case where $\LT(\bfh) \,{<}\, \LM(t_{i}\bfg_{i}) \,{=}\, \LM(t_{j}\bfg_{j})$.
  In this case, there exists a term $t$ such that $t_{i} = t \frac{\mm{i,j}}{\mm{i}}$ and
  $t_{j} = ct \frac{\mm{i,j}}{\mm{j}}$, and $\bfh = t \SPol(\bfg_{i},\bfg_{j})$.
  If the pair $(\bfg_{i},\bfg_{j})$ is non-singular, then by hypothesis, $\SPol(\bfg_{i},\bfg_{j}) =
  \SGPol(\bfg_{i},\bfg_{j})$ $\redpol$-reduces to $0$ modulo $\GG$, and so does $\bfh = t
  \SPol(\bfg_{i},\bfg_{j})$.
  If the pair is singular, then $\sig(\bfh) \precneq t\sigS(\bfg_{i},\bfg_{j}) \preceq \bfT$, and by hypothesis, $\bfh$ $\redpol$-reduces to $0$ modulo $\GG$.
\end{proof}

\begin{proof}[Proof of Th.~\ref{prop:correctness-Lichtblau}]
  \newcommand{\tp}{\tau}
\newcommand{\cp}{\chi} 
  Write $\GG = \{\bfg_{1},\dots,\bfg_{r}\}$ and $\GG_{z} = \{\bfz_{1},\dots,\bfz_{s}\}$.
  For all $i,j$, let $\cc{i} = \LC(\bfg_{i})$, $\mm{i} = \lm(\bfg_{i})$ and $\mm{i,j} = \lcmlm(\bfg_{i},\bfg_{j})$.
  Let $\bfh \in \II$ with signature $\bfT$ be such that $\bfh$ is not $\redpol$-reducible modulo $\GG$.
  Assume that $\bfT$ is minimal for this property, and that among such elements of $\II$ with signature $\bfT$, $\lm(\bfh)$ is minimal.

  Consider a decomposition of $\bfh$ with respect to $\GG$,
  \begin{equation}
    \label{eq:4}
    \textstyle\bfh = \sum_{u=1}^{k} \tp_{u}\bfg_{i_{u}} + \sum_{v=1}^{l} t^{(z)}_{v} \bfz_{j_{v}}, 
  \end{equation}
  such that for all $u,v$, 
  $\LT(\tp_{u}\bfg_{i_{u}}) \geq \LT(\tp_{u+1}\bfg_{i_{u+1}})$,
  $\max(\tp_{u}\sig(\bfg_{i_{u}})) \preceq \bfT$,
  $\max(t^{(z)}_{v}\sig(\bfz_{j_{v}})) \preceq \bfT$
  and $t^{(z)}_{v}\sig(\bfz_{j_{v}}) \geq t^{(z)}_{v+1}\sig(\bfz_{j_{v+1}})$.
  Such a representation exists, by definition of the signature of $\bfh$ and the first hypothesis on $\GG$.
  Assume that, among such representation, this one is minimal in the sense that $\LM(\tp_{1}\bfg_{i_{1}})$ is minimal,
  and the largest $j$ such that $\LM(\tp_{j}\bfg_{i_{j}}) = \LM(\tp_{1}\bfg_{i_{1}})$ is minimal for this property.
  For all $u$, let $\cp_{u} = \lc(\tp_{u})$.

  \paragraph{Case 1: $\bfh$ is not a syzygy} 
  We want to prove that $\LM(\tp_{1}\bfg_{i_{1}}) > \LM(\tp_{2}\bfg_{i_{2}})$.
  It will in particular imply that $\lt(\bfh) = \lt(\tp_{1}\bfg_{i_{1}})$, and thus that $\bfh$ is $\redpol$-reducible modulo $\GG$.
  By minimality of $\lm(\bfh)$, this will prove that $\bfh$ $\redpol$-reduces to 0 modulo $\GG$.

In order to reach a contradiction, assume that $\LM(\tp_{1}\bfg_{i_{1}}) {=}
  \LM(\tp_{1}\bfg_{i_{2}})$.
  By definition of the least common multiplier, there exists $m \in \Mon(A)$ such that
  $\LM(\tp_{1}\bfg_{i_{1}}) = \LM(\tp_{2}\bfg_{i_{2}}) = m \,\mm{i_{1},i_{2}}$.

  If $\cc{i_{1}}$ divides $\cc{i_{2}}$ or $\cc{i_{2}}$ divides $\cc{i_{1}}$, then by
  Lemma~\ref{lemme:1}, and expanding the s-reductions, $\tp_{1}\bfg_{i_{1}} + \tp_{2}\bfg_{i_{2}}$ admits a standard
  Sig-representation, which can be substituted in the
  representation~\eqref{eq:4}, contradicting minimality.
  
  For the other case, by assumption the G-polynomial of $\bfg_{i_{1}}$ and $\bfg_{i_{2}}$ is $\redpol$-reducible modulo $G$, so there exists $\bfg_{i_{3}} \in G$ such that
  \vspace{-0.05cm}
  \begin{enumerate}[a.]
    \item $\cc{i_{1}} \mm{i_{1},i_{2}}$ is divisible by $\lt(\bfg_{i_{3}})$, say, $t'_{1} \lt(\bfg_{i_{3}}) = \cc{i_{1}}
    \mm{i_{1},i_{2}}$;
    \item\label{item:bound_sig_g3} $t'_{1} \sig(\bfg_{i_{3}}) \preceq \sigS(\bfg_{{i}_{1}},\bfg_{i_{2}})$;
    \item $\cc{i_{2}} \mm{i_{1},i_{2}}$ is divisible by $\lt(\bfg_{i_{3}})$, say, $t'_{2} \lt(\bfg_{i_{3}}) = \cc{i_{2}}
    \mm{i_{1},i_{2}}$;
    \item $t'_{2} \sig(\bfg_{i_{3}}) \preceq \sigS(\bfg_{i_{1}},\bfg_{i_{2}})$.
  \end{enumerate}
  In particular, $\cc{i_{3}}$ divides $\cc{i_{1}}$, say, $\cc{i_{1}} = a'_{1} \cc{i_{3}}$.
  So the SG-polynomial of $\bfg_{i_{1}}$ and $\bfg_{i_{3}}$ is an S-polynomial, and by
  Lemma~\ref{lemme:1}, it $\redpol$-reduces to $0$ modulo $\GG$.
  So it admits a standard Sig-representation
  \begin{equation}
    \label{eq:5}
    \SGPol(\bfg_{i_{1}},\bfg_{i_{3}})
    = \frac{\mm{i_{1},i_{3}}}{\mm{i_{1}}} \bfg_{i_{1}}
    - a_{1} \frac{\mm{i_{1},i_{3}}}{\mm{i_{3}}}\bfg_{i_{3}}
    = \sum_{j\geq 1} t^{(1)}_{j} \bfg^{(1)}_{i_{j}} + \sum \text{syz.},
  \end{equation}
  where $\sum \text{syz.}$ is a linear combination of elements of $\GG_{z}$.
  So
  \begin{equation}
    \label{eq:11}
    \frac{\mm{i_{1},i_{3}}}{\mm{i_{1}}} \bfg_{i_{1}}
    =  a_{1}  \frac{\mm{i_{1},i_{3}}}{\mm{i_{3}}}\bfg_{i_{3}}
    + \sum_{j\geq 1} t^{(1)}_{j} \bfg^{(1)}_{i_{j}} + \sum \text{syz.},
  \end{equation}
  with $\LT(t^{(1)}_{j}\bfg^{(1)}_{i_{j}}) \leq \LT(\SPol(\bfg_{i_{1}},\bfg_{i_{3}})) < \mm{i_{1},i_{3}}$.
  Since $\mm{i_{1},i_{2}}$ is divisible by $\mm{i_{3}}$, $\mm{i_{1},i_{2}}$ is divisible by
  $\mm{i_{1},i_{3}}$, say, $\mm{i_{1},i_{2}} = \mu_{1} \mm{i_{1},i_{3}}$.
  So 
  \begin{align}
    \label{eq:12}
    \tp_{1}\bfg_{i_{1}} &= \cp_{1}m\mu_{1} \frac{\mm{i_{1},i_{3}}}{\mm{i_{1}}}\bfg_{i_{1}} \\
    &=  \cp_{1} z_{1} m\mu_{1} \frac{\mm{i_{1},i_{3}}}{\mm{i_{3}}}\bfg_{i_{3}}
    + \sum_{j\geq 1} \cp_{1} z_{1} m \mu_{1} t^{(1)}_{j} \bfg^{(1)}_{i_{j}}  + \sum \text{syz.}
  \end{align}
  and it is a standard representation of $\tp_{1}\bfg_{1}$.
  Furthermore, since the signature of $\bfg_{i_{3}}$ is bounded (property \ref{item:bound_sig_g3}), it is also a standard Sig-representation.
  Similarly, there exists $z_{2}$, $\mu_{2}$ and $(t^{(2)}_{j}, i_{j})$ such that
  \begin{align}
    \label{eq:13}
    \tp_{2}\bfg_{i_{2}} 
    &=  \cp_{2} a_{2} m\mu_{2} \frac{\mm{i_{2},i_{3}}}{\mm{i_{3}}}\bfg_{i_{3}}
    + \sum_{j\geq 1} \cp_{2} a_{2} m \mu_{2} t^{(2)}_{j} \bfg^{(2)}_{i_{j}}  + \sum \text{syz.}
  \end{align}
  and it is a standard Sig-representation.
  We can group both representations together, and obtain a standard Sig-representation of
  $\tp_{1}\bfg_{i_{1}}+\tp_{2}\bfg_{i_{2}}$
  \begin{align}
    \label{eq:14}
    \tp_{1}\bfg_{i_{1}}+\tp_{2}\bfg_{i_{2}} &=
      \cp_{1} a_{1} m\mu_{1} \frac{\mm{i_{1},i_{3}}}{\mm{i_{3}}}\bfg_{i_{3}}
      + \cp_{2} a_{2} m\mu_{2} \frac{\mm{i_{2},i_{3}}}{\mm{i_{3}}}\bfg_{i_{3}}
      + \sum \dots \\
    &= \big( \cp_{1} a_{1}
     + \cp_{2} a_{2} \big) \frac{\mm{i_{1},i_{2}}}{\mm{i_{3}}}
    m \bfg_{i_{3}} + \sum \dots 
  \end{align}
  which, substituted into~\eqref{eq:4}, contradicts the minimality assumption.
  
  \paragraph{Case 2: $\bfh$ is a syzygy}
  The same proof as above, if $\lt(\bfh)=0$, implies that $k=0$ in  \eqref{eq:4}.
  Now, consider, among all decompositions of the form \eqref{eq:4}, one where $j = \max(v : t^{(z)}_{v}\sig(\bfz_{j_{v}})) \simeq t^{(z)}_{1}\sig(\bfz_{j_{1}})$ is minimal.
  Assume that $j > 0$, and thus $t^{(z)}_{1}\sig(\bfz_{j_{1}}) \simeq t^{(z)}_{2}\sig(\bfz_{j_{2}})$.
  Then, by assumption, the $\sigG$-comb. of $\bfz_{1}$ and $\bfz_{2}$ is sig-reducible by $\GG_{z}$, thus, there exists $t'_{1} \in \Ter(A)$ and $j'_{1} \in \{1,\dots,s\}$ such that
  $t'_{1}\sig(\bfz_{j'_{1}}) = t^{(z)}_{1}\sig(\bfz_{j_{1}}) + t^{(z)}_{2}\sig(\bfz_{j_{2}})$, and so
  $\mathbf{z} := t^{(z)}_{1}\bfz_{1} + t^{(z)}_{2}\bfz_{2} - t'_{1}\sig(\bfz_{j'_{1}})$ has signature $\precneq \bfT$.
  So subtracting $\bfz$ from the decomposition~\eqref{eq:4} results in a decomposition with fewer terms matching the signature of $\bfh$, contradicting the minimality of~$j$.
\end{proof}

\begin{table}
  \centering
\caption{Comparison between Algo.~\ref{algo:KRK-like} and Algo.~\ref{algo:Lichtblau}}
\label{tab:pairs}
\vspace{-0.3cm}
\small
\begin{tabu}{crrrr}
  \toprule
\multirow{2}{*}{\textbf{System}} & \multicolumn{2}{c}{Algo~\ref{algo:KRK-like}} & \multicolumn{2}{c}{Algo.~\ref{algo:Lichtblau}} \\
& Pairs/red./to 0 & Time & Pairs/red./to 0 & Time\\
\midrule
\textbf{Katsura-4} & 420/188/0 & 1.35 & 855/412/0 & 1.60\\
\textbf{Katsura-5} & 2048/723/0 & 32.40 & 7178/3983/0 & 79.87\\
\textbf{Cyclic-5} & 221/63/0 & 0.37 & 347/158/0 & 0.71\\
\textbf{Cyclic-6} & 3019/742/8 & 200.33 & 9672/5782/8 & 616.82\\
\bottomrule
\end{tabu}
\vspace{-0.2cm}
\end{table}

\newcommand{\pbox}[2][c]{\begin{table}{#1}#2\end{table}}

\begin{table*}
  \centering
  \caption{Comparative timings for module computations with the signature-based algorithms and with Magma (in seconds)}
\label{tab:timings}
\vspace{-0.3cm}
\small
  \begin{tabu}{crrrcrrr}
  \toprule
  \multirow{2}{*}{\textbf{System}} & \multicolumn{3}{c}{With signatures (Algo. \ref{algo:KRK-like})} &\hspace{5pt}& \multicolumn{3}{c}{Magma} \\
  & {Sig-GB} & {Recons.} & {\textbf{Total}} && {GB} & {\textbf{GB with coordinates}} & {\textbf{Module of syzygies}} \\
\midrule
\textbf{Cyclic-5} & 0.4 & 0.1 & \textbf{0.5} &  & 0.01 & \textbf{954.6} & \textbf{954.8}\\
\textbf{Cyclic-6} & 200.3 & 10.6 & \textbf{210.9} &  & 2.08 & \textbf{>24h} & \textbf{>24h}\\
\bottomrule
\end{tabu}
\vspace{-0.2cm}

\end{table*}

\section{Algorithms in practice}
\label{sec:algorithms-practice}

\subsection{Further optimizations}
\label{sec:furth-optim}

In this section, we briefly describe additional criteria which can be used to eliminate elements.
Firstly, in both algorithms, one can use Buchberger's coprime and chain criteria (either as-is or using Gebauer and Möller's implementation).
Buchberger's criterion does not require any modification to work with signatures, whereas the chain criterion needs to ensure that the signature of the pairs used to discard the redundant one is small enough \cite{GerdtHashemi-2013}.
Note that in all cases, we need to consider terms (with their coefficients) and not just monomials.
For Lichtblau's algorithm, refined versions of those criteria relaxing the condition on the coefficients have been described in~\cite{Lichtblau}  and can also be used here.

We have already stated that some criteria can be used to eliminate pairs based on their signatures.
We have also already cited the F5 criterion, which allows to pre-populate the basis of syzygies with the signatures of predictable syzygies, as well as the possibility of discarding G-polynomials which are sig-reducible by $\GG_{z}$, or which are super reducible and s-reducible by $\bfg \in \GG$.
Similarly, in Lichtblau's algorithm, one can discard any S-polynomial which is super reducible by $\GG$.
A natural question is whether the cover criterion would allow to systematically discard G-polynomials in Algorithm~\ref{algo:KRK-like}, or to discard new elements in Algorithm~\ref{algo:Lichtblau} (including non-regular S-polynomials).
Experimentally, it appears that indeed, most such elements which are covered can be discarded without impacting the correctness of the algorithm.

Another point which can have a large impact on the complexity is the choice of the order of the pairs, \emph{i.e. }, how to break ties between elements with signatures which are $\simeq$.
A strategy which seems to yield good results over $\mathbb{Z}$ is to compare the absolute value of the coefficient of the signatures, so as to create super reducers and covering candidates sooner.

We have only mentioned top reductions, that is, reductions of the leading coefficient, but as usual, the definitions generalize to allow reductions of the rest of the terms.
Finally, we have only defined reductions where the leading coefficient of the reducer divides the coefficient to be reduced.
In some rings, and in particular in the case of Euclidean rings, it is also possible to perform modular reductions on the coefficients without impacting the correctness of the result.
This significantly improves the performances of the algorithm.
The same can be done for sig-reductions.

\subsection{Experimental data}
\label{sec:experimental-data}

We have written a prototype implementation of both algorithms in \textsf{Magma}\footnote{\url{https://gitlab.com/thibaut.verron/signature-groebner-rings}}, for the PoT ordering and $R=\ZZ$.
We report in Table~\ref{tab:pairs} data on the number of pairs being processed, reduced and reduced to zero for different benchmark systems (Katsura-$n$ and Cyclic-$n$), as well as indicative computation times.
In practice, it appears that Algo.~\ref{algo:KRK-like} is more efficient than Algo.~\ref{algo:Lichtblau}, both in terms of number of computed pairs and in time.
This appears to be due to the relaxed restrictions allowing non-singular polynomials, more than the lack of criteria.


Our prototype implementation of both algorithms of this paper is slower than Magma's implementation of F4~\cite{faugere:1999:F4} over $\ZZ$ for merely computing Gröbner bases.
As mentioned earlier, the computation of signatures also allows to compute the coefficients of the elements of the Gröbner basis in terms of the input, and a basis of the module of syzygies, by performing and tracking s-reductions~\cite{Gao-2015-new-framework-for}.
The process only depends on the definition of a Sig-GB and a Sig-basis of syzygies, and therefore works in our setting as well.
In Table~\ref{tab:timings}, we give the computation time for this reconstruction using  Algo.~\ref{algo:KRK-like}, as well as comparable routines in \textsf{Magma}\footnote{\textsf{Groebner} of an \textsf{IdealWithFixedBasis} for a GB with coordinates, and \textsf{SyzygyMatrix} for a basis of the syzygy module.}. 
In several instances, we observe that the use of signatures gives a significant speed-up for those computations.

One particularity of signature-based algorithms over rings is that, due to the partial order on the signatures, they typically compute a large number of elements with incomparable signatures.
This problem does not appear over fields, and future work will focus on ways to eliminate more of those elements, or to speed-up the computations at a given signature (for instance using linear algebra techniques similar to F4).




\end{document}